\documentclass[journal]{IEEEtran}
\usepackage{amsmath}
\usepackage{amsfonts}%
\usepackage{amssymb}%
\usepackage{amsthm}
\usepackage{graphicx}
\usepackage{epstopdf}
\usepackage{mathrsfs}
\usepackage{cite}
\usepackage{tikz}
\usepackage{xcolor}
\usepackage[colorlinks=true,linkcolor=blue,citecolor=red]{hyperref}
\usepackage{setspace}
\usepackage{algorithm}
\usepackage{algpseudocode}
\usepackage{pgfplots}

\ifCLASSINFOpdf
\else
\fi

\newtheorem{theorem}{Theorem}

\newtheorem{lemma}{Lemma}

\theoremstyle{definition}
\newtheorem{remark}{Remark}

\def\ba{\left[\begin{array}}
\def\ea{\end{array}\right]}
\def\as{\overset{a.s}{\longrightarrow}}

\newcommand{\mb}[1]{
\mathbf{#1}
}
\newcommand{\Trf}[1]{
\text{Tr}\left(#1\right)
}

\def\tSINR{\text{SINR}}

\def\be{\begin{equation}}
\def\ee{\end{equation}}
\def\bea{\begin{align}}
\def\eea{\end{align}}
\def\bean{\begin{align*}}
\def\eean{\end{align*}}
\def\noi{\noindent}
\begin{document}
%
\title{Optimal Power Allocation and User Loading for Multiuser MISO Channels with Regularized Channel Inversion
\thanks{A part of this paper was presented at the Asilomar Conference on Signals, Systems and Computers, Pacific Grove, CA, USA, November 2011, \cite{Muharar_asil11}.}}
%
%
%

\author{Rusdha~Muharar,
        Randa~Zakhour,~\IEEEmembership{Member,~IEEE,}
        and~Jamie~Evans,~\IEEEmembership{Member,~IEEE}
\thanks{Rusdha Muharar is with the Department of Electrical and Computer Systems Engineering, Monash University, Clayton, VIC 3800, Australia  (e-mail: rusdha.muharar@monash.edu).}
\thanks{Randa Zakhour is part-time faculty at the Schools of Engineering of the American University of Science and Technology and of the Lebanese International University, Lebanon (email: randa.zakhour@gmail.com).}
\thanks{Jamie Evans is with the Department of Electrical and Computer Systems Engineering, Monash University, Clayton, VIC 3800, Australia (email: jamie.evans@monash.edu).}}

\maketitle

\begin{abstract}
We consider a multiuser system where a single transmitter equipped with multiple antennas (the base station) communicates with multiple users each with a single antenna. Regularized channel inversion is employed as the precoding strategy at the base station. Within this scenario we are interested in the problems of power allocation and user admission control  so as to maximize the system throughput, i.e., which users should we communicate with and what power should we use for each of the admitted users so as to get the highest sum rate. This is in general a very difficult problem but we do two things to allow some progress to be made. Firstly we consider the large system regime where the number of antennas at the base station is large along with the number of users. Secondly we cluster the downlink path gains of users into a finite number of groups. By doing this we are able to show that the optimal power allocation under an average transmit power constraint follows the well-known water filling scheme. We also investigate the user admission problem which reduces in the large system regime to optimization of the user loading in the system.
\end{abstract}


%

\section{Introduction}
%
%
%
%
Multi-Input Multi-Output (MIMO) technologies are currently being adopted in many wireless communication standards such as fourth generation (4G) cellular networks. In multiuser MIMO downlink transmissions or broadcast channels (MIMO-BC), the capacity region was 
characterized in \cite{Weingarten_it06} and is achieved by employing Dirty Paper Coding (DPC) at the transmitter.  
However, implementing this technique in practice is computationally expensive \cite{Erez_it05,Spencer_commag04}.
Multiuser beamforming techniques such Zero Forcing (ZF) and 
Regularized Channel Inversion (RCI) are sub-optimal in term of the sum-rate but offer a lower complexity in the implementation. 
ZF can asymptotically achieve a sum rate that is close to that of DPC by appropriate power allocation and user scheduling \cite{Yoo_jsac06}. 

In Multi-Input Single-Output (MISO) broadcast channels, finding the optimal power allocation policy maximizing the sum rate for ZF is a convex optimization problem and has the solution that follows the water-filling (WF) scheme, see e.g., \cite{Yoo_jsac06}.  
In contrast, the optimal power allocation 
for the RCI precoder is a non-convex optimization problem with many local optima \cite{Xu_eurasip08,Xu_globecom07,Jin_icc09}, 
even in the case of all users having the same path-loss. 
In \cite{Xu_eurasip08,Xu_globecom07}, the authors investigated the sum rate maximization of  
MIMO broadcast channels with RCI under a total power constraint. 
They showed that the problem is a global difference of convex functions (d.c.) optimization problem and proposed the local gradient method 
to solve the problem. Their numerical results suggest that employing an RCI precoder with power allocation 
gives a better sum rate compared to the ZF. Reference \cite{Jin_icc09} extends the previous works, 
but in the MISO broadcast channels setting, by putting additional quality of service (QoS) constraints 
where each user's data rate should be above a specified minimum rate. The authors re-cast the 
optimization problem as a series of geometric programming (GP) problems, called iterative GP (IGP).   

As already mentioned, besides power allocation, selecting the users for transmission can improve the system performance. 
It has been shown in \cite{Yoo_jsac06} that a combination of water-filling  based power allocation 
and a user selection scheme, called semi-orthogonal user selection (SUS), in MISO BC systems with ZF precoder can approach the sum rate obtained by employing DPC 
when the number of users is large. A similar conclusion is also presented in \cite{Dimic_sp05, Wang_sp08} but by using greedy search algorithms for the 
user selection. The performance analysis of that algorithm for the case of finite (at most two) scheduled users was carried out in \cite{Ozyurt_comletter12}.
The authors in \cite{Sjobergh_icassp08} also proposed a greedy user selection for the RCI precoder: their algorithm is based on the closed form approximation of the expected sum rate. 
In \cite{Dai_jsac08}, Dai et al. studied MISO BC systems with ZF precoder under a finite-rate or quantized feedback. 
The proposed power allocation scheme is binary or on/off. They showed that the feedback rate and the received 
SNR affect the optimal number of active ('on') users. Moreover, their scheme can be applied in 
heterogeneous environments where the users may have different path-losses. A similar problem is also considered in \cite{Zhang_comm11,Zhang_isit09} but with different settings.
Besides considering the finite-rate feedback, the authors take into account the feedback delay by using a Gauss-Markov model; they  also assume a homogeneous environment and equal power allocation across the users. 
A sum rate approximation expression as a function of the number of users is derived. 
As a result, the number of users can be adjusted adaptively based on the feedback delay and channel quantization error (or feedback rate).  
This strategy is similar to the multi-mode transmission scheme considered in this paper.

In this paper, we will be considering power and user loading (also called group loading) allocations, in addition to regularization parameter optimization, in a  MISO BC with heterogeneous users. As mentioned earlier, solving the optimal power allocation problem alone is a challenging task \cite{Xu_eurasip08,Xu_globecom07,Jin_icc09}.
Adding the user loading allocation and the regularization parameter of the precoder into the optimization problem increases the complexity of the task.
To tackle the problem, we therefore apply two simplifying strategies. Firstly, we consider the large system regime where the number of users, $K$, and the number of transmit antennas of the BS, $N$, tend to  infinity with a fixed ratio. In this large system, we show how many of the related problems simplify and key insights can be obtained. Secondly, we divide all users into a finite number of groups or clusters, where all users in each group have approximately the same distance from the BS and therefore share the same distance dependent path-loss.

As a result of applying the first strategy, we are able to show that in the large system regime, each user's SINR tends to a deterministic quantity, called the limiting SINR: limiting SINRs depend only on users' allocated power and path-loss and not on the realization of the fast fading coefficients. Following that, under the second modeling strategy, we consider the joint optimization of the users powers and the regularization parameter.  For a fixed regularization parameter, we show that the optimal power allocation problem under an average power constraint is \textit{convex} and
the \textit{optimal power allocation follows the familiar water-filling strategy} \cite{Goldsmith_it97,Yoo_jsac06}. Similar results were also obtained in \cite{Wagner_it12} but with different approaches in the large system analysis. By substituting back the power allocation scheme to
the limiting sum rate expression, we can derive the optimal regularization parameter . Even though it does not yield a closed form expression, this substitution does result a one dimensional optimization problem which can be solved by standard line search algorithms.

It should be noted that the water-filling scheme may allocate zero power to some of the groups. Consequently, one may ask whether it is better to include the channel states of those groups in the precoder or not. This leads to the second part of the paper where we consider a multi-mode transmission scheme (see also \cite{Zhang_comm11,Zhang_isit09}). In this scheme, for a given total number of groups ($L$) and group loading of each group, we determine the optimal number of groups for the transmission and also which groups  the BS should communicate with. We arrange or sort the groups based on their path-losses in a descending order. We investigate two cases. In the first case, for each group, the BS can only decide between transmitting to all the users in the group or to none of them. 
We consider  a uniform group loading over the groups. In the mode $m$ transmission where the BS only communicates to $m$ groups (out of $L$), it is optimal for the BS to transmit to the first $m \leq L$ groups. The optimal mode can then be determined by comparing the maximum limiting sum rate of each mode.
In the second case, the BS is allowed to communicate with any subset  of the users in a group. We provide a necessary condition for the optimal group-loading allocation for each group. Assuming that $M \leq L$ groups are allocated positive power, the group loadings of the first $M-1$ groups should be set at their maximum value and the group loading for the $M$-th group can be in between zero and its maximum value. 
We also propose an algorithm to solve this optimization problem. Considering the group loading allocation, the algorithm offers a lower complexity  in comparison to  brute force search methods. In both cases, the optimal power allocation and regularization parameter are also considered.

Throughout the paper, the following notations are used. $\mathbb{E}[\cdot]$ denotes the statistical expectation and $\as$ refers to almost sure convergence. $\frac{\partial f}{\partial x}$ denotes the partial derivative of $f$ with respect to (w.r.t.) $x$  and $\frac{\partial f(x^\star)}{\partial x}$ represents $\frac{\partial f(x)}{\partial x}$ at $x=x^\star$. The circularly symmetric complex Gaussian vector with mean $\pmb{\mu}$ and covariance matrix $\mathbf{\Sigma}$ is denoted by $\mathcal{CN}(\pmb{\mu},\mathbf{\Sigma})$. $|a|$ denotes the magnitude of the complex variable $a$. $\|\cdot\|$ represents the Euclidean norm. $\succeq$ represents  element-wise inequality for the vectors. $\Trf{\cdot}$ denotes the trace of a matrix. $\mb{I}_N$ and $\mathbf{0}_N$ denote an $N\times N$ identity matrix and a $1\times N$ zero entries vector, respectively. $(\cdot)^T$ and $(\cdot)^H$ refer to the transpose and Hermitian transpose, respectively.  LHS and RHS refer to the left-hand and right-hand side of an equation, respectively.
 
\section{System Model}
\subsection{Finite-size system model}
We consider a MISO broadcast channel with an RCI precoder at the 
transmitter end. The base station has $N$ antennas and each of $K$ users is 
equipped with a single antenna. The propagation channel coefficient between 
 transmit antenna $n \in \{1,\dots,N\}$ and user $k \in \{1,\dots,K\}$ is 
denoted by $h_{k,n}$. Thus, the channel gain vector between the BS and user $ 
k$ is represented by the row vector $\mb{h}_k=[h_{k,1}, h_{k,2},\dots,h_{k,N}]
 \in \mathbb{C}^N$. It is assumed that the entries of $\mb{h}_k$ are i.i.d. and $\mb{h}_k \sim \mathcal{CN}(\mb{0},\mb{I}_N)$\footnote{
Even though here, we assume a specific distribution for $\mb{h}$, 
the large system analysis holds for any distribution of $\mb{h}_k$ if the entries of 
$\frac{1}{\sqrt{N}}\mb{h}_k$ are i.i.d. with zero mean, variance $\frac{1}{N}$ and have finite eighth moment (see e.g. \cite{Couillet_book11}).}.

We model the data symbol vector as 
$\mathbf{s}=\mathbf{\Lambda}^{1/2}\bar{\mathbf{s}}$, where $\bar{\mathbf{s}}$ is the normalized (power) data symbol vector, i.e., 
$\mathbb{E}[\bar{\mathbf{s}}\bar{\mathbf{s}}^H]=\mathbf{I}_K$. Let $\mathbf{\Lambda}=\text{diag}(p_1, p_2,\cdots, p_K)$ where  
$p_k$ denotes the power allocated to  user $k$. The transmitted data vector can written as
$\mathbf{x}=\mb{P}\mathbf{\Lambda}^{1/2}\bar{\mathbf{s}}$ and has a power
constraint $\mathbb{E}[\|\mathbf{x}\|_2^2]=P_d$. Let $\mb{H}=[\mb{h}_1\ \mb{h}_2\ \cdots\ \mb{h}_K]^T$ be the channel gain matrix. 
The  RCI precoder matrix, $\mb{P}$, takes the form  $\mb{P}=c \left(\mb{H}^H\mb{H} + \alpha \mb{I}_N \right)^{-1}\mb{H}^H$,
where $\alpha$ is the regularization parameter that controls the amount of interference introduced to the users and $c$ is the normalizing constant chosen to meet the transmit power constraint $\mathbb{E}[\|\mathbf{x}\|_2^2]=P_d$,  that is,
\be\label{eq:c2_power}
c^2=\frac{P_d}{\Trf{\mathbf{\Lambda}\mathbf{H}(\mathbf{H}^H\mathbf{H}+\alpha \mathbf{I}_N)^{-2}\mathbf{H}^H}}.
\ee 
The received signal for user $k$ is given by
\begin{align*}
y_k&=a_k \mathbf{h}_{k}\mb{x} + w_k \\
&= ca_k\sqrt{p_k}\mathbf{h}_k(\mathbf{H}^H\mathbf{H}+\alpha \mathbf{I}_N)^{-1}\mathbf{h}^H_k\bar{s}_k \\
& \qquad + \sum_{j\neq k}^K ca_k\sqrt{p_j}\mathbf{h}_k(\mathbf{H}^H\mathbf{H}+\alpha \mathbf{I}_N)^{-1}\mathbf{h}^H_j\bar{s}_j,
\end{align*}
where $a_k^2$ denotes the slow-varying path-loss between the base station and the receiver of user $k$. 
Therefore, the SINR attained by user $k$ can be expressed as
\begin{align}\label{eq:finite_sinr_power}
\tSINR_k=\frac{c^2a_k^2p_k |\mb{h}_k(\mb{H}^H\mb{H} + \alpha\mathbf{I}_N)^{-1}\mb{h}^H_k|^2}{\sum_{j\neq k}^K c^2a_k^2 p_j |\mathbf{h}_k(\mathbf{H}^H\mathbf{H}+\alpha \mathbf{I}_N)^{-1}\mathbf{h}^H_j|^2+\sigma^2}.
\end{align}
It is clear that the $\tSINR_k$ is a random quantity since it depends on the propagation channels that fluctuate randomly. In the large system limit, as we will see in the next section, this randomness disappears.


\subsection{Large-system regime SINR}
The following theorem provides the convergence of the $\tSINR_k$  \eqref{eq:finite_sinr_power} when the system dimensions, that is, $K$ and $N$, grow large with their ratio fixed. Note that this result also follows from \cite[Corollary 2]{Wagner_it12}.
\begin{theorem}\label{th:lim_SINR}
Let $\rho=\frac{\alpha}{N}$ be the normalized regularization parameter and $g(\beta,\rho)$ be the solution of $g(\beta,\rho)=\left(\rho + \frac{\beta}{1+g(\beta,\rho)}\right)^{-1}$. Let ${\cal P}=\displaystyle \lim_{K\to \infty} \frac{1}{K}\sum_{k=1}^K p_k$. Suppose that the limit ${\cal P}$ exists and is bounded. Then, as $K,N \to \infty$ with $\frac{K}{N}\to \beta$, $\tSINR_k$  \eqref{eq:finite_sinr_power}
converges almost surely to a deterministic quantity, $\tSINR_k^\infty$, given by
\begin{align}
\tSINR_k^\infty &= \bar{p}_k g(\beta,\rho)\frac{\gamma_k+\frac{\gamma_k\rho}{\beta}(1+g(\beta,\rho))^2}{\gamma_k+(1+g(\beta,\rho))^2},
\label{eq:lim_SINR_power} 
\end{align}
where $\gamma_k=\frac{P_d a_k^2}{\sigma^2}$ is defined as the effective SNR and $\bar{p}_k=\frac{p_k}{\cal P}$ is the normalized power w.r.t. ${\cal P}$.
\end{theorem}
\begin{IEEEproof}
Refer to Appendix \ref{proofLSA}.
\end{IEEEproof}

We call $\tSINR_k^\infty$ the limiting SINR of user $k$. Note that it is different for each user and depends on $a_k$ and $p_k$.
Let $f_k(\beta,\rho)$ be the RHS of \eqref{eq:lim_SINR_power} excluding $\bar{p}_k$. Then, we can write \eqref{eq:lim_SINR_power} as 
\be\label{eq:lim_SINR_power2}
\tSINR_k^\infty = \bar{p}_kf_k(\beta,\rho).
\ee
Note that $f_k$ also depends on $a_k$ and $P_d/\sigma^2$ via $\gamma_k$ but these are assumed to be fixed throughout this paper. It is also obvious that $f_k$ is independent of $\bar{p}_k$. This property will ease the analysis in finding the power allocation that maximizes the (limiting) sum rate in the next section. 

\section[Optimal Power allocation and regularization parameter]{Optimal Power allocation and regularization parameter}\label{sec:power_alloc}
Let us consider the following scenario. We divide all $K$ users into $L$  groups, where $L$ is finite. All users in each group are assumed to have the same path-loss. Here, we assume that $a_1 \geq a_2 \geq \dots \geq a_L$. The number of users in group $j$ is denoted by $K_j$, with $\sum_{j=1}^L K_j=K$.
We also assume that $K_j$ and $N$ tend to be large with a fixed ratio $\beta_j=\frac{K_j}{N}$. It represents the user or group loading of group $j$. 
Since the path-loss and other parameters $\beta,\rho$ as well as SNR are the same for all users in a group,  
then based on \eqref{eq:lim_SINR_power}, we can assume that the  power allocated to each user in that group is also the same. 
This assumption holds for the rest of the paper. 

Based on the above scenario, we can define the limiting achievable sum rate per antenna as follows
\begin{align}\label{eq:lim_sumrate_power}
R_\text{sum}^\infty=\sum_{j=1}^L \beta_j\log\left(1+\text{SINR}_j^\infty\right).
\end{align}  

Our goal in this section is to find the optimal power allocation that maximizes  $R_\text{sum}^\infty$.
Moreover, it is also interesting to explore how the regularization parameter of the RCI precoder adapts to different path-losses and also user powers.    
A joint optimization problem can be formulated as follows, 
%
\begin{eqnarray}
\mathbf{P1}:  &\underset{\bar{\mathbf{p}}\succeq \mb{0},\rho\geq 0}{\text{max.}} & R_\text{sum}^\infty \nonumber\\
&\text{s.t.} & \displaystyle \sum_{j=1}^L \frac{\beta_{j}}{\beta}\bar{p}_j=\frac{1}{\beta}\pmb{\beta}^T{\bar{\mathbf{p}}}\leq 1. \label{const:P1_avPower}
\label{eq:opt_prob}
\end{eqnarray}
In the above, we use lowercase bold letters to denote column vectors with size $L$, e.g., $\bar{\mb{p}}=[\bar{p}_1, \bar{p}_2, \dots, \bar{p}_L]^T$. This notation will be used for the rest of this paper, unless otherwise stated.      
Note that the constraint \eqref{const:P1_avPower} can be considered as the large system average power constraint. $\mathbf{P1}$ also requires $\bar{\mathbf{p}}$ and $\rho$ to be non-negative.

Before addressing the solution of $\mathbf{P1}$, we characterize the objective function as a function of $\bar{p}_j$. 
Let $R_{\text{sum},j}^\infty=\beta_j\log\left(1+\text{SINR}_j^\infty\right)$ denote the sum rate for group $j$. It can be checked that  it is an increasing function in $p_j$. Moreover, we can show that the following lemma holds.
\begin{lemma}\label{lemma:conc_p}
The sum rate per antenna $R_\text{sum}^\infty$ is concave in $\bar{\mathbf{p}}$.
\end{lemma}
\begin{proof}
The second derivative of the limiting SINR w.r.t. $\bar{p}_j$ is 
\[
\frac{\partial^2 \text{SINR}_j^\infty}{\partial \bar{p}_j^2}=-\frac{f^2_j(\beta,\rho)}{(1+\bar{p}_jf_j(\beta,\rho))^2} < 0.
\]
This implies that $\text{SINR}_j^\infty$ is concave in $\bar{p}_j$. Since the $\log$ operation does not change the concavity, therefore $R_{\text{sum},j}^{\infty}$ is also concave in $\bar{p}_j$. Moreover, $R_\text{sum}^\infty$ is a linear combination of $R_{\text{sum},j}^\infty$ and this operation preserves the concavity.   
\end{proof}

From the lemma above, we can see that for a fixed $\rho$, $\mathbf{P1}$ is a convex program because $-R_{\text{sum}}^\infty$ is convex in $\bar{\mathbf{p}}$ and the constraints are linear.
For a fixed $\bar{\mathbf{p}}$, $\tSINR_k^\infty$ is not concave in $\rho$ but quasi-concave \cite{RusdhaPrep}. Since $\log$ is a non-decreasing function then $R_{\text{sum},j}^\infty$ is also quasi-concave (not concave) in $\rho$. Since a linear combination operation does not necessarily preserve the quasi-concavity, the sum rate need not be quasi-concave. 


Now, let us consider the Lagrangian for $\mb{P1}$, as stated below\footnote{For notational simplicity, we use $\mathcal{L}$ to denote the Lagrangian $\mathcal{L}(\mb{x},\pmb{\lambda})$, where $\mb{x}$ and $\pmb{\lambda}$ are the optimizing variables and the Lagrange multipliers, respectively.}
\[ 
\mathcal{L} = \sum_{j=1}^L \beta_j \log(1+\bar{p}_jf_j(\beta,\rho)) - \lambda\left(\frac{1}{\beta}\pmb{\beta}^T\bar{\mb{p}}-1\right) + \pmb{\xi}^T\bar{\mb{p}} +  \kappa \rho,
\]
where $\lambda$ and $\pmb{\xi}$ are the Lagrange multipliers for the average power and non-negative power constraints respectively, and $\kappa$ is the Lagrange multiplier for the constraint $\rho \geq 0$. 
Let $\bar{\mb{p}}^\star,\rho^\star$ be the solutions for $\mb{P1}$. At these points, the associated Karush-Kuhn-Tucker (KKT) optimality conditions are
\begin{align}
\frac{\partial \mathcal{L}}{\partial \bar{p}_j} &= \beta_j\left(\frac{f_j(\beta,\rho^\star)}{1+\bar{p}_j^\star f_j(\beta,\rho^\star)}-\lambda\right) + \xi_j = 0 \label{eq:stat_pj_power_P1}\\
\frac{\partial \mathcal{L}}{\partial \rho} &= \sum_{j=1}^L \frac{\beta_j \bar{p}_j^\star}{1+\bar{p}_j^\star f_j(\beta,\rho)}\frac{\partial f_j(\beta,\rho^\star)}{\partial \rho} +\kappa = 0, \label{eq:stat_rho_power_P1} 
\end{align} 
and
\begin{align}
	&\frac{1}{\beta}\pmb{\beta}^T\bar{\mb{p}}^\star \leq 1, \lambda\left(\frac{1}{\beta}\pmb{\beta}^T\bar{\mb{p}}^\star-1\right) = 0, \lambda \geq 0,\label{eq:KKT_P1_power1}\\
	&\bar{\mb{p}}^\star\succeq 0, \xi_j\bar{p}_j^\star=0, \xi_j \geq 0, \ j=1,\dots,L, \label{eq:KKT_P1_power2}\\
	&\rho^\star \geq 0, \kappa \rho^\star=0, \kappa \geq 0.  \label{eq:KKT_P1_rho}
\end{align}

Recall that for a given $\rho$, $\mb{P1}$ reduces to a convex program. In this case, it is easy to show that the KKT conditions \eqref{eq:stat_pj_power_P1}, \eqref{eq:KKT_P1_power1} and \eqref{eq:KKT_P1_power2} lead to the optimal power allocation strategy maximizing the limiting sum rate, as presented in the following theorem.     

\begin{theorem}\label{Th:power_alloc}
For a fixed $\rho$, the optimal power allocation for the optimization problem $\mathbf{P1}$ follows the water-filling (WF) scheme and is given by
\begin{equation}
\bar{p}_j^\star=\left[\frac{1}{\lambda} -\frac{1}{f_j(\beta,\rho)}\right]_+
\label{eq:power_alloc}
\end{equation}
where $[x]_+=\max(0,x)$. The constant (Lagrange multiplier) $\lambda$ is the solution of
\[
\sum_{j=1}^L \beta_j\bar{p}_j^\star = \beta,
\]
for which the average power constraint is satisfied with equality.
\end{theorem}

In the WF scheme above, $1/\lambda$ can be perceived as the water level. It determines how  power is poured to each user and is based on the value of $f_j(\beta,\rho)$. Recall that the limiting SINR is given by $\bar{p}_j^\star f_j(\beta,\rho)$. It can be checked that $f_j(\beta,\rho)$ is increasing in $a_j$. Thus, more power will be allocated for the users with better channels which can be represented by the path losses $\{a_j\}$. 
Note that in this case, fairness amongst  users  could be an issue since some users might have zero rate. 

\begin{remark}
To find $\lambda$ we can follow the following steps (see also \cite{Cui_sp07}). Since we assume $a_1 \geq a_2 \geq \dots \geq a_L$,  then $\bar{p}_1^\star \geq \bar{p}_2^\star \geq \dots \geq \bar{p}_L^\star $. Now let us assume that the first $m$ groups have non-zero power. To determine $\lambda$, we just need to solve $\sum_{j=1}^m \beta_j\bar{p}_j^\star=\beta$. 
Using $\bar{p}_j^\star$ in \eqref{eq:power_alloc}, it is easy to show that
\[ 
\lambda=\dfrac{\sum_{j=1}^m \beta_j}{\beta+\sum_{j=1}^m \frac{\beta_j}{f_j(\beta,\rho)}}.
\]
The power allocated to group $j$ is then given by
\[ 
\bar{p}_j^\star = \frac{\beta+\sum_{j=1}^m \frac{\beta_j}{f_j(\beta,\rho)}}{\sum_{j=1}^m \beta_j}-\frac{1}{f_j(\beta,\rho)}
\]
To determine $m$, we just need to find $m$ such that $\bar{p}_m^\star > 0$ and $\bar{p}_{m+1}^\star \leq 0$. \hfill\IEEEQEDclosed
\end{remark}

  By using the KKT optimality conditions above, the optimal $\rho^\star$ can be found as stated in the following theorem.
\begin{theorem}\label{th:opt_rho}
Let $\bar{\mb{p}}^\star$ be as in \eqref{eq:power_alloc}. The maximum limiting sum rate, $R_\text{sum}^\infty$, is obtained by choosing $\rho^\star$ that  satisfies
\be\label{eq:opt_rho1_power}
\sum_{j=1}^L \frac{\beta_j \bar{p}_j^\star f_j^2(\beta,\rho^\star)}{1+\bar{p}_j^\star f_j(\beta,\rho^\star)}\left(\frac{\rho^\star}{\beta}-\frac{1}{\gamma_j}\right)=0.
\ee
and it is bounded by
\be\label{eq:interval_opt_rho_power}
\frac{\beta}{\gamma_1} \leq \rho^\star \leq \frac{\beta}{\gamma_L}. 
\ee
\end{theorem}
\begin{proof}
See Appendix \ref{app:opt_rho}.
\end{proof}

Note that by using \eqref{eq:power_alloc} in \eqref{eq:opt_rho1_power}, it is straightforward to see that \eqref{eq:opt_rho1_power} becomes a one-dimensional zero/root-finding problem. Thus, the optimal $\rho$ can be found by using existing line search algorithms for the interval given in \eqref{eq:interval_opt_rho_power}. 
\begin{figure}
\vspace{0.5cm}
\centering
\begin{tikzpicture}
		
	\begin{axis}[compat=newest, small,
		xmin=0, xmax=21,
		width=8.5 cm,
		xlabel={$\frac{P_d}{\sigma^2}$ (dB)},
		ylabel={$\mathbb{E}[R_\text{sum}]$},
		grid=major,
		legend pos=south east,
	]
	\pgfplotsset{
		every mark/.append style={solid}
		}
	\addplot table {Rsum_p_FS.dat};
	\addplot[mark=x] table {Rsum_p_LS.dat};
	\legend{\small $\bar{\mb{p}}^\star_\text{\tiny FS}$, \small $\bar{\mb{p}}^\star$};
	\end{axis}
\end{tikzpicture}
\caption{Comparison of the average sum rate between using $\bar{\mathbf{p}}=\bar{\mathbf{p}}^{\star}_\text{\tiny FS}$ and $\bar{\mathbf{p}}=\bar{\mathbf{p}}^\star$ for $L=2$, $\beta=1$, $N=8$, $\beta_j=1/2$, $\rho=\rho^\star$ and $a_j^2=1/j^2$.} 
\label{fig:Rsum_FS_vs_LS}
\end{figure}
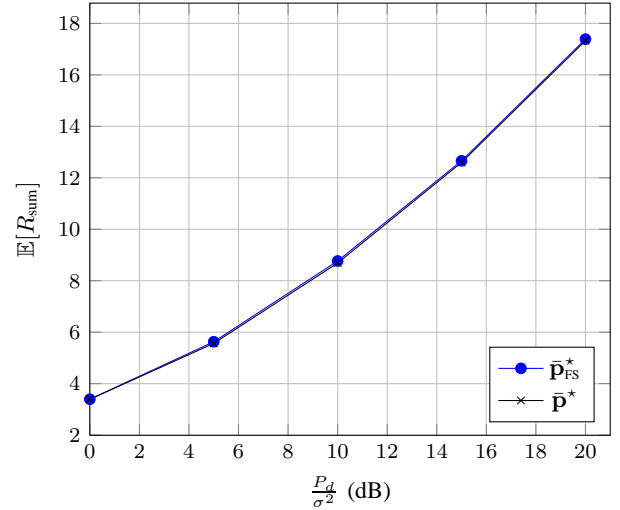

Figure \ref{fig:Rsum_FS_vs_LS} illustrates the validity of using the large system results for the finite size system. 
We generate 500 channel realizations and for each realization we compute the optimal power allocation, denoted by $\bar{\mathbf{p}}^*_\text{\tiny FS}$, by a grid search. In the plot, we compare the average sum rate, denoted by $\mathbb{E}[R_\text{sum}]$, between using the power allocation $\bar{\mathbf{p}}$ in \eqref{eq:power_alloc} and $\bar{\mathbf{p}}^*_\text{\tiny FS}$. The gap between the curves in the figure is very small and can be said negligible. 

\section{Multimode Broadcast Channels}\label{s:multimode}

In the previous sections, we considered the optimal power allocation that maximizes the (limiting) sum rate where the base station (BS) communicates to all $L$  groups simultaneously. In that setting, the channels from all groups of users were included in the precoding matrix, even those allocated zero power (i.e. not actually served). This can be seen as an inefficiency and  leads us to ask: how much do we stand to gain if we take care of this inefficiency?
As an illustration, let us consider the case of $L=3$. We set the group loading for each group to be uniform i.e., $\beta_j=\beta/L$. 
Figure \ref{fig:multimode} shows the limiting sum rate obtained when the BS communicates to only the first $m \leq L$ groups, denoted by $R_\text{sum}^{(m),\infty}$.
This means that we only include the channel of the users from these $m$ groups in the system model and in designing the precoder.  We call this scheme mode-$m$ transmission. The figure demonstrates that for some values of cell-loading $\beta$, the maximum sum rate is achieved when $m < L$. The simulation also shows that the optimal $m$ changes with $\beta$: we call this scheme \emph{multi-mode transmission}.

In multimode transmission, it is clear that there are ${{L}\choose{m}}$ combinations of the groups that can be chosen by the base station to communicate with. The question is which mode and group combination that will give the highest sum rate? Intuition would suggest that if the base station is communicating with $m$ groups then we would choose the $m$ groups with the strongest channel gains in order to maximize the sum rate. We would then need test at most $L$ mode and group combinations. Below we show that this intuition is indeed correct for different assumptions on $\beta_j$ although the proof is non-trivial.  

\subsection{Binary Group Loading}
In this section, we investigate the following optimization problem:
\begin{equation*}\label{eq:opt_prob_2}
\begin{array}{ccl}
\mathbf{P2}: & \underset{\bar{\mathbf{p}} \succeq \mb{0},\pmb{\beta},\rho,\beta \geq 0}{\text{max.}} & R_\text{sum}^\infty \\
&\text{s.t.} & \dfrac{1}{\beta}\pmb{\beta}^T\bar{\mb{p}} \leq 1\\ 
& & \dfrac{1}{\beta}\pmb{\beta}^T\mb{1} = 1\\ 
& &  \beta_j \in \{0,\beta_{j,\max}\},
\end{array}
\end{equation*}
where $\mb{1}$ is a column vector with all 1 entries. It can be seen that $\mathbf{P2}$ is similar to $\mathbf{P1}$, but with additional design variables: $\pmb{\beta}, \beta$ and additional constraints related to them. In $\mathbf{P2}$, $\beta_j$ is only allowed to have value either 0 or $\beta_{j,\max}$ and we call this scheme \textit{binary user loading allocation}. Therefore, $\beta_j$ will determine whether the BS  transmits to users in group $j$ or not. The latter occurs when $\beta_j=0$. In that case, the channel gain matrix of the users in group $j$ is not included in the precoder design.  

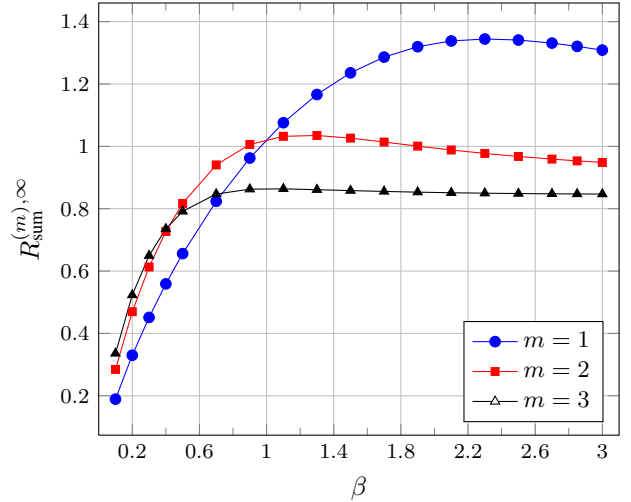
\begin{figure}[t]
\vspace{0.5cm}
\centering
\begin{tikzpicture}
		
	\begin{axis}[small,
		compat=newest,
		xmin=0, xmax=3.1,
		width=8.5 cm,
		xlabel={$\beta$},
		ylabel={$R_\text{sum}^{(m),\infty}$},
		xtick={0.2,0.6,1,1.4,1.8,2.2,2.6,3},
		minor x tick num=1,
		grid=major,
		legend pos=south east,
	]
	\pgfplotsset{
		every mark/.append style={solid}
		}
	\addplot table[x=beta,y=mode1] {Rsum_modes_L3.dat};
	\addplot[color=red, mark=square*, mark size=1.5] table[x=beta,y=mode2] {Rsum_modes_L3.dat};
	\addplot[mark=triangle*] table[x=beta,y=mode3] {Rsum_modes_L3.dat};
	\legend{\small $m=1$, \small $m=2$, \small $m=3$};
	\end{axis}
\end{tikzpicture}
\caption{Multimode Transmission for $L=3$, $\beta_j=\beta/L$ and $a_j^2=1/j^2$ .}
\label{fig:multimode}
\end{figure}

First, let us investigate the optimal strategies for $\mathbf{P2}$ when $\beta_{j,\max}$ is the same for all groups, i.e., $\beta_{j,\max}=\breve{\beta}$. 
Let us consider the mode-$m$ transmission. In that case, we have $m$ groups with $\beta_j=\breve{\beta}$ and the remaining groups have $\beta_j=0$. 
Let $\mathcal{G}\subset \{1,2,\dots, L\}$, $|\mathcal{G}|=m$ be the set of the group indexes that the BS communicates to ($\beta_j > 0, j\in \mathcal{G}$). 
Then, the maximum limiting sum rate achieved for a given $\mathcal{G}$ can be obtained by solving 
\be\label{eq:opt_mode_m}
\begin{array}{cl}
 \underset{\bar{\mathbf{p}}\succeq \mb{0},\rho}{\text{max.}} & \displaystyle R_\text{sum}^{(m),\infty}(\mathcal{G}) = \sum_{j \in \mathcal{G}} \beta_j\log(1+\bar{p_j}f_j(\beta,\rho))  \\
\text{s.t.} & \displaystyle \frac {1}{m} \sum_{j \in \mathcal{G}} \bar{p}_j\leq 1 .
\end{array}
\ee
We should note that in the average power constraint we use the fact that the total group loading $\beta$ is  $\sum_{j\in \mathcal{G}} \beta_j=m\breve{\beta}$. We can also see that  \eqref{eq:opt_mode_m} is equivalent to $\mathbf{P1}$.  Thus, its solutions can be obtained by using the same strategies as in solving $\mathbf{P1}$. The maximum limiting sum rate for mode-$m$ transmission can be attained by evaluating \eqref{eq:opt_mode_m} for every possible choice of group combinations $\mathcal{G}$, i.e.,
\be\label{eq:opt_mode_m_2}
\breve{R}_\text{sum}^{(m),\infty} = \underset{\mathcal{G}\subset \{1,\dots,L\}, |\mathcal{G}|=m}{\text{max.}} \  \displaystyle R_\text{sum}^{(m),\infty}(\mathcal{G}).
\ee
By using the formulation \eqref{eq:opt_mode_m_2}, we can rewrite $\mathbf{P2}$ as
\be\label{eq:opt_binary_betaj}
\mathbf{P2}:  \underset{m \leq L}{\text{max.}} \  \breve{R}_\text{sum}^{(m),\infty} 
\ee
As mentioned earlier, for \eqref{eq:opt_mode_m_2} there are ${{L}\choose{m}}$ possible choices or candidates for the optimal $\mathcal{G}$. For the problem  \eqref{eq:opt_binary_betaj}, the number of candidates becomes $\sum_{i=1}^L {{L}\choose{i}}=2^L$. In the following Lemma we show that \eqref{eq:opt_mode_m_2} has the intuitively obvious solution mentioned above thereby reducing the number of candidate mode/group combinations for \eqref{eq:opt_binary_betaj} to $L$.

\begin{lemma}\label{lemma:multimode_equalbeta}
$\breve{R}_\text{sum}^{(m),\infty}$ is achieved by choosing $\mathcal{G}=\mathcal{G}^\star$ where $\mathcal{G}^\star=\{1,2,\dots,m\}$.
\end{lemma}
\begin{proof}
Let $\mathcal{G}^\star=\{1,2,\dots,m\}$. Also, let $\mathcal{S}\subset\{1,\dots,L\}$ with $|\mathcal{S}|=m$ such that $\mathcal{G}^\star \neq \mathcal{S}$. Moreover, the elements of $\mathcal{S}$ are arranged in an increasing order. Let $\mathbf{a}_{\mathcal{G}^\star}$ and $\mathbf{a}_{\mathcal{S}}$ be the path-loss gain vector for group combinations $\mathcal{G}^\star$ and $\mathcal{S}$, respectively. It is clear that $\mathbf{a}_{\mathcal{G}^\star} \succeq \mathbf{a}_{\mathcal{S}}$. Thus, for a fixed power and regularization parameter, it follows that
 $R_\text{sum}^{(m),\infty}(\mathcal{G}^\star) \geq R_\text{sum}^{(m),\infty}(\mathcal{S})$. Now, suppose that $\bar{\mathbf{p}}^\star_\mathcal{S}$ and $\rho^\star_\mathcal{S}$ are the optimal power allocation and regularization parameter under $\mathcal{S}$. Let us denote the corresponding limiting sum rate as $R_\text{sum}^{(m),\infty}(\mathcal{S},\rho^\star_\mathcal{S},\bar{\mathbf{p}}^\star_\mathcal{S})$. Under $\mathcal{G}^\star$, let us choose
$\bar{\mathbf{p}}_\mathcal{G}=\bar{\mathbf{p}}^\star_\mathcal{S}$ and $\rho_{\mathcal{G}^\star}=\rho^\star_\mathcal{S}$ for the power allocation and $\rho$, respectively. Even though those choices are not optimal in maximizing  $R_\text{sum}^{(m),\infty}(\mathcal{G}^\star)$, they satisfy the constraint in \eqref{eq:opt_mode_m}. Since both $\mathcal{G}^\star$ and $\mathcal{S}$ have the same allocations for power and $\rho$, then 
it follows that $R_\text{sum}^{(m),\infty}(\mathcal{G}^\star,\rho^\star_\mathcal{S},\bar{\mathbf{p}}^\star_\mathcal{S})\geq R_\text{sum}^{(m),\infty}(\mathcal{S},\rho^\star_\mathcal{S},\bar{\mathbf{p}}^\star_\mathcal{S})$. This concludes the proof.
\end{proof}
It is clear from the lemma above that we greatly reduce the complexity of $\mathbf{P2}$. Now, we only need to compare $L$ limiting sum rates, $\breve{R}_\text{sum}^{(m),\infty}$. It is also easy to see that Lemma \ref{lemma:multimode_equalbeta} also holds when $\beta_{1,\max} \geq \beta_{2,\max} \geq \dots \geq \beta_{L,\max}$. 
For a more general setup, we can relax the last constraint of $\mb{P2}$ so that  $0\leq\beta_j\leq\beta_{j,\max}$. 
This will be addressed in the following section.

\subsection{Fractional Group Loading}
In this section, we consider a fractional group loading scheme where $\beta_j$ can take values in $[0,\beta_{j,\max}]$.
This allows the BS to transmit not to all the users in the groups but some of them. In this case, $\mathbf{P2}$ becomes
\begin{equation*}
\begin{array}{ccl}
\mathbf{P3}: & \underset{\bar{\mathbf{p}}\succeq \mb{0},\pmb{\beta},\rho,\beta\geq 0}{\text{max.}} & R_\text{sum}^\infty \\
&\text{s.t.} & \dfrac{1}{\beta}\pmb{\beta}^T\bar{\mb{p}} \leq 1\\ 
& & \dfrac{1}{\beta}\pmb{\beta}^T\mb{1} = 1\\ 
& & \mb{0}\preceq \pmb{\beta}\preceq \pmb{\beta}_{\max}.
\end{array}
\label{eq:opt_prob_3}
\end{equation*}
To find the solution for $\mathbf{P3}$, we start by writing the Lagrangian of $\mathbf{P3}$ as follows:
\begin{align*}
	\mathcal{L} &= \sum_{j=1}^L \beta_j \log(1+\bar{p}_jf_j(\beta,\rho)) - \lambda\left( \frac{1}{\beta}\pmb{\beta}^T\bar{\mb{p}}-1)\right) + \pmb{\xi}^T\bar{\mb{p}}   \\
	& \quad   + \mu\left(\frac{1}{\beta}\pmb{\beta}^T\mb{1} -1\right) +  \pmb{\nu}\pmb{\beta} -  \pmb{\eta}^T(\pmb{\beta}-\pmb{\beta}_{\max}) + \kappa \rho + \eta \beta,
\end{align*}
where $\lambda, \kappa, \mu, \eta, \pmb{\xi}, \pmb{\nu},\pmb{\eta}$ are the Lagrange multipliers for the constraints of $\mathbf{P3}$. Let $\bar{\mb{p}}^\star,\pmb{\beta}^\star,\rho^\star,\beta^\star$ be the (candidate) solutions for $\mb{P3}$. The KKT necessary optimality conditions are 
\begin{align}
\frac{\partial \mathcal{L}}{\partial \rho} &= \sum_{j=1}^L \frac{\beta_j^\star \bar{p}_j^\star}{1+\bar{p}_j^\star f_j(\beta^\star,\rho^\star)}\frac{\partial f_j(\beta^\star,\rho^\star)}{\partial \rho} +\kappa = 0 \label{eq:stat_rho_power} \\
\frac{\partial \mathcal{L}}{\partial \bar{p}_j} &= \beta_j^\star\left(\frac{f_j(\beta^\star,\rho^\star)}{1+\bar{p}_j^\star f_j(\beta^\star,\rho^\star)}-\lambda\right) + \xi_j = 0 \label{eq:stat_pj_power}\\
\frac{\partial\mathcal{L}}{\partial \beta_j} &= \log(1+\bar{p}_j^\star f_j(\beta^\star,\rho^\star)) - \lambda(\bar{p}_j^\star-1) \notag\\
&\qquad\qquad\qquad\qquad\qquad\quad + \nu_j -\eta_j + \mu =0 \label{eq:stat_betaj_power}\\
\frac{\partial \mathcal{L}}{\partial \beta} &= \sum_{j=1}^L \frac{\beta_j^\star \bar{p}_j^\star}{1+\bar{p}_j^\star f_j(\beta^\star,\rho^\star)}\frac{\partial f_j(\beta^\star,\rho^\star)}{\partial \beta} - \mu + \eta = 0 \label{eq:stat_beta_power}
\end{align} 
and
\begin{gather*}
	\lambda\left(\pmb{\beta}^{\star T}\bar{\mb{p}}^\star - \beta^\star\right) = 0, \pmb{\beta}^{\star T}\bar{\mb{p}}^\star - \beta^\star\leq 0,  
		\bar{\mb{p}}^\star \succeq \mb{0},\xi_j\bar{p}_j^\star=0,  \\
	 \nu_j\beta_j^\star=0, \eta_j(\beta_j^\star-\beta_{j,\max}) = 0,\\
	  \mb{0}\preceq \pmb{\beta}^\star \preceq \pmb{\beta}_{\max}, \pmb{\beta}^{\star T}\mb{1}-\beta^\star=0,  \\
	\beta^\star \geq 0, \eta \beta^\star=0,  \rho^\star\geq 0, \kappa \rho^\star=0, \\
	[\lambda\ \kappa\  \eta]^T \succeq 0, \pmb{\xi} \succeq  0,   \pmb{\nu}\succeq 0, \pmb{\eta} \succeq 0,
\end{gather*}
for all $j=1,\dots,L$. 

Let us consider the stationarity condition \eqref{eq:stat_rho_power}. In solving $\mathbf{P1}$, we have shown that $f_j(\beta,\rho)$ is increasing in $\rho$ up to $\rho=\beta/\gamma_j$ and then decreasing. Thus, the optimal $\rho$ can not be zero ($\kappa=0$) and at the optimum,
\be\label{eq:opt_rho_power_P3}
\sum_{j=1}^m \frac{\beta_j^\star \bar{p}_j^\star}{1+\bar{p}_j^\star f_j(\beta^\star,\rho^\star)}\frac{\partial f_j(\beta^\star,\rho^\star)}{\partial \rho}=0. 
\ee
Looking at \eqref{eq:stat_pj_power}, one can see that when $\bar{p}_j > 0$ ($\xi_j=0$), it satisfies
\[
\bar{p}_j^\star = \left[\frac{1}{\lambda} - \frac{1}{f_j(\beta^\star,\rho^\star)}\right]_+
\]
which has a similar form to the solution for $\mathbf{P1}$. Since $a_1\geq\ \dots\ \geq a_L $, then $\bar{p}_1^\star\geq\dots\geq\bar{p}_L^\star$. At the optimum, the following holds
\[
\sum_{j=1}^L\beta_j^\star\left(\left[\frac{1}{\lambda} - \frac{1}{f_j(\beta^\star,\rho^\star)}\right]_+-1\right)=0
\]
and it can be used to determine $\lambda$.
	
Exploring the stationary condition \eqref{eq:stat_betaj_power} will lead us to the following result.
\renewcommand{\labelenumi}{(\roman{enumi})} 
\begin{lemma}\label{lemma:betaj_alloc}
	The optimal $\{\beta_j\}$ allocation is such that
	\begin{enumerate}
	\item the first $M$ groups, for some $M\leq L$, will be allocated non-zero power,
	\item $\beta_1^\star,\beta_2^\star,\dots,\beta_{M-1}^\star$ are all at the maximum possible values,
	\item $0\leq\beta_M^\star\leq \beta_{M,\max}$,
	\item the remaining groups are allocated zero power.
	\end{enumerate}	
\end{lemma} 
\begin{proof}
See Appendix \ref{proofbetaj_alloc}.
\end{proof}

 We should note that in the lemma above, we do not know the optimal value of $M$ maximizing the limiting sum rate since there are several values of $M$ that satisfy the lemma. Let $R_\text{sum}^{(i),\infty}$ be the achieved limiting sum rate with $M=i$. Let $\mathcal{M}=\{1,2,\dots,L\}$ be the set of possible values for $M$. Then, the optimal $M$ is given by
\be\label{eq:opt_mode_varbeta}
M^\star=\underset{i \in \mathcal{M}}{\arg \max} \ R_\text{sum}^{(i),\infty}\ .
\ee  
\noi We should note that in evaluating $R_\text{sum}^{(i),\infty}$, we use $\{\beta_j^\star\}$ allocation scheme in Lemma \ref{lemma:betaj_alloc}, $\beta^\star=\sum_{j=1}^i \beta_j^\star$ and also the stationary conditions  in \eqref{eq:stat_rho_power} and \eqref{eq:stat_pj_power} to determine $\rho^\star$ and $\bar{\mb{p}}^\star$ respectively. The value for $\beta_M^\star$ must satisfy \eqref{eq:stat_betaj_power} with $\nu_M=0$ and $\eta_M=0$, i.e.,
\be\label{eq:betaM}
\log(1+\bar{p}_M^\star f_M(\beta^\star,\rho^\star)) - \lambda(\bar{p}_M^\star-1) + \mu =0,
\ee
where $\mu$ is given by \eqref{eq:mu_at_opt}. Thus, solving \eqref{eq:opt_mode_varbeta} correspondingly solves $\mathbf{P3}$. The steps in solving it are presented in Algorithm \ref{alg:solve_P3}. 
\begin{algorithm*}[t]
\caption{\quad Algorithm for Solving $\mathbf{P3}$ } \label{alg:solve_P3}
\begin{algorithmic}[1] 
\State $\mathcal{M}=\{\}$ \Comment{Contain possible values for $M$} 
\For{$j=1$ \textbf{to} $L$}
			\State $\beta_i^\star=\beta_{i,\max}, \forall\ i=1,\dots,j$     \Comment{Assume that $\beta_j^\star=\beta_{j,\max}$}
			\State $\lambda, \rho^\star, [\bar{p}_1^\star\ \dots \bar{p}^\star_j]^T \gets $ Solving $\mathbf{P1}$ with $\beta^\star= \sum_{i}^j \beta_i^\star$
			\State Determine $M$ s.t. $\bar{p}_{M}^\star>0$ and  $\bar{p}_{M+1}^\star = \dots = \bar{p}_j^\star =0$ \Comment{$M\geq 1$}
			\If{$M \in \mathcal{M}$}
					\State \textbf{continue} \Comment{Skip the remaining steps and go to the next iteration (Step 2)}
			\EndIf
			\State $\mathcal{M} \gets M$
			\State $\beta_{M + 1}^\star=\dots=\beta_j^\star=0$	
			\State Compute $\mu$ according to \eqref{eq:mu_at_opt}
			\State $\eta_M = \log(1+\bar{p}_M^\star f_M(\beta^\star,\rho^\star)) - \lambda(\bar{p}_M^\star-1) +  \mu$
			\If{$\eta_M < 0$}
					\State $\beta_M^\star \in [0,\beta_{M,\max}] \gets $ Solving $\mathbf{P1}$ and \eqref{eq:betaM} with $\beta^\star=\sum_{i}^{M-1} \beta_{i,\max} + \beta_M^\star$
			\EndIf
			\State Compute $R_\text{sum}^{(M),\infty}$ with the updated $\beta$ and $\{\beta_j\}$
\EndFor
\State $M^\star \gets $  Solving \eqref{eq:opt_mode_varbeta} 
\end{algorithmic}
\end{algorithm*}

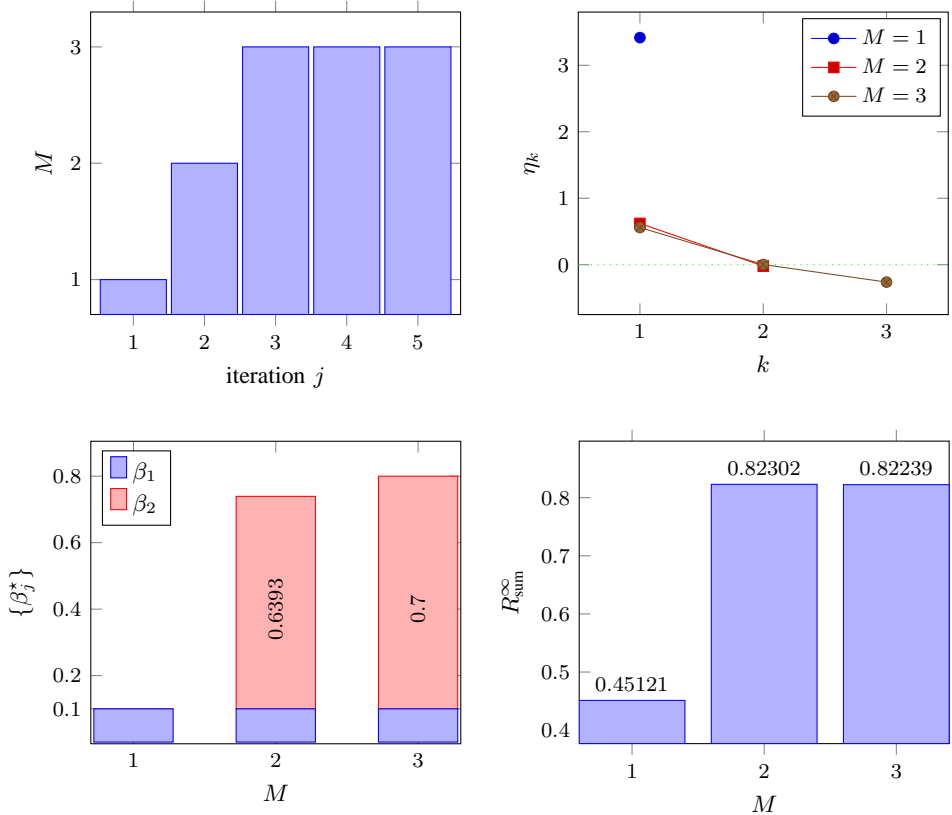
\begin{figure*}[!tbh]
\vspace{0.5cm}
\centering
\begin{tabular}{rl}
\begin{tikzpicture}[baseline,trim axis left]
	\begin{axis}[compat=newest,small, enlargelimits=0.15, ybar=0pt, bar width=25, xlabel={iteration $j$}, ylabel={$M$}, ytick=data
	]
	\addplot coordinates {(1,1) (2,2) (3,3) (4,3) (5,3)};
	\end{axis}
\end{tikzpicture}
&\hspace{0.18cm}
\begin{tikzpicture}[baseline, trim axis right]
	\begin{axis}[compat=newest, small, enlargelimits=false, xlabel={$k$}, ylabel={$\eta_k$}, xtick={1,2,3}, ymin=-0.75, ymax=3.8
	]
	\addplot coordinates {(1,3.4158)};
	\addplot coordinates {(1,0.6207) (2,-0.0194)};
	\addplot coordinates {(1,0.5591) (2,0.0028) (3,-0.2614)};
	\addplot[color=green,dotted] coordinates {(0.5,0)  (3.5,0)};
	\legend{\footnotesize $M=1$,\footnotesize $M=2$,\footnotesize $M=3$}
	\end{axis}
\end{tikzpicture}\\[40pt]
\begin{tikzpicture}[baseline, trim axis left]
	\begin{axis}[compat=newest, small, enlargelimits=0.15, ybar stacked, bar width=30, xlabel={$M$}, ylabel={$\{\beta_j^\star\}$}, xtick=data, ytick={0.1,0.2,0.4,0.6,0.8},
	legend pos=north west
	]
	\node[rotate=90] at (axis cs:2,0.4) {\footnotesize $0.6393$};
	\node[rotate=90] at (axis cs:3,0.4) {\footnotesize $0.7$};
	\addplot coordinates {(1,0.1) (2,0.1) (3,0.1)};
	\addplot coordinates {(1,0) (2,0.6393) (3,0.7)};
	\legend{\small $\beta_1$, \small $\beta_2$};
	\end{axis}
	
\end{tikzpicture}
& 
\begin{tikzpicture}[baseline, trim axis right]
	\begin{axis}[compat=newest, small, enlargelimits=0.2,  ybar=0pt!, bar width=40, xlabel={$M$}, ylabel={$R_\text{sum}^\infty$}, xtick=data,
	]
	\addplot coordinates {(1,0.45121) (2,0.82302) (3,0.82239)};
	\node[above] at (axis cs:1,0.45121) {\footnotesize $0.45121$};
	\node[above] at (axis cs:2,0.82302) {\footnotesize $0.82302$};
	\node[above] at (axis cs:3,0.82302) {\footnotesize $0.82239$};
	\end{axis}
\end{tikzpicture}
\end{tabular}
\caption{Algorithm \ref{alg:solve_P3} implementation for $L=5$, $\pmb{\beta}_{\max}=[0.1\ 0.7\ 0.1\ 0.05\ 0.05 ]^T$,  $a_j^2=1/j^2$ and $P_d/\sigma^2 = 10$ dB. }
\label{fig:multimode_var_beta}
\end{figure*}

 We have $L$ iterations when in a particular iteration, say iteration $j$, the first $j$ groups are considered. Assuming those group to have their group loading at the maximum value, the corresponding optimal power allocation (i.e., solving $\mathbf{P1}$) is computed. Then, the value of $M \leq j$ for that iteration can be determined by using the fact that $\bar{p}_{M+1}^\star=0$. We should note that different $j$s may give the same $M$ and hence, we need only to consider one of them. After having $M$, we can set $\beta_{M+1}^\star=\dots=\beta_j^\star=0$. To determine the optimal value for $\beta_M^\star$, we need to compute $\eta_M$.
If $\eta_M > 0$, $\beta_M^\star=\beta_{M,\max}$ (we already set this in the first step). Otherwise, $0 \leq \beta_M^\star \leq \beta_{M,\max}$. In the latter case, we need to solve $\mathbf{P1}$ and \eqref{eq:betaM} simultaneously. Then, we can update the value for $\{\beta_j^\star\}_{j=1}^M$ and $\beta^\star$ and also compute the corresponding limiting sum rate. In the final steps, we compare the limiting sum rates for different $M$ and the maximum is the solution of $\mathbf{P3}$.


Figure \ref{fig:multimode_var_beta} illustrates the implementation results of algorithm \ref{alg:solve_P3} for the case:
$L=5$, $a_j^2=1/j^2,\ j=1,\dots,L$, $\pmb{\beta}_{\max}=[0.1\ 0.7\ 0.1\ 0.05\ 0.05 ]^T$ where the $j$-th element corresponds to $\beta_{j,\max}$ and $P_d/\sigma^2 = 10$ dB. From the (upper-left) plot, we can see that we only have three possible values for $M$, i.e., $\mathcal{M}=\{1,2,3\}$. For $M=1$, we have a positive $\eta_M$ while for $M=2$ and $M=3$, $\eta_M$ is negative. We should note that for $M=3$, $\eta_2$ is slightly above zero (0.0028). Executing step 14 in the algorithm \ref{alg:solve_P3} yields $\beta_2^\star=0.6393$ and $\beta_3^\star=0$ for $M=2$ and $M=3$, respectively. Even though $M=2$ and $M=3$ have the same two groups with positive group loading, they have different total group loadings, i.e., $0.7393$ and $0.8$, respectively and consequently different sum rates. The last plot in the bottom-right shows that the maximum limiting sum rate is achieve when $M=2$. To validate the result from Algorithm \ref{alg:solve_P3}, we perform a grid search where $\beta$ takes values between 0 and 1 with 0.001 increment. For each value of $\beta$, the corresponding limiting sum rate is computed. The results are plotted in Figure \ref{fig:Rsum_grid_analytic}. The plot shows that the maximum limiting sum rates and the optimal $\beta$ obtained from the grid search and Algorithm \ref{alg:solve_P3} are identical. This confirms our theoretical analysis and the proposed algorithm. We should note that even though the line around the optimal $\beta$ looks flat, closer inspection of the numerical values of the limiting sum rates in that region reveals that the limiting sum rate is actually increasing until reaching the optimal $\beta$ and then decreasing. 

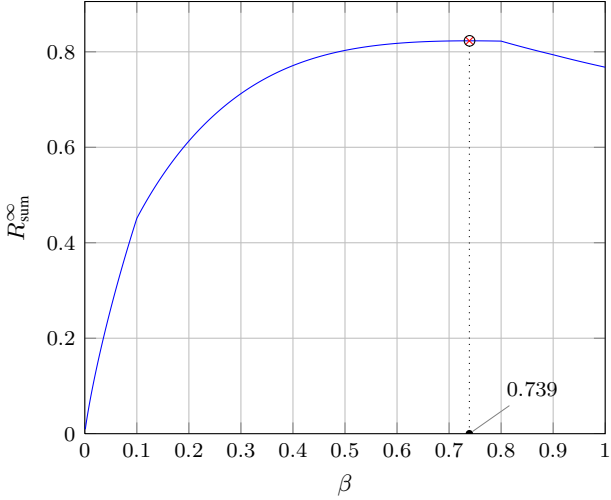
\begin{figure}[t]
\vspace{0.2cm}
\centering
\begin{tikzpicture}
	\tikzset{
		small dot/.style={fill=black,circle,scale=0.3}
	}
	\begin{axis}[compat=newest, small,
		xmin=0, xmax=1, ymin=0,
		width=8.5 cm,
		xlabel={$\beta$},
		ylabel={$R_\text{sum}^\infty$},
		grid=major
	]
	\addplot[blue,-] table {Rsum_exhaustive.dat};
	\addplot[color=red, mark=x] coordinates {(0.739,0.8230)};
	\addplot[color=black, mark=o] coordinates {(0.7393,0.8230)};
	\addplot[dotted,-] coordinates {(0.739,0) 
	(0.739, 0.8230)
	};
	\node[small dot,pin=45:{\footnotesize $0.739$}] at (axis description cs:0.739,0) {};
	\end{axis}
\end{tikzpicture}
\caption{The maximum limiting sum rate obtained from the grid search ($\textcolor[rgb]{1,0,0}{\times}$) and Algorithm \ref{alg:solve_P3} ($\circ$).}
\label{fig:Rsum_grid_analytic}
\end{figure}

We can also observe from the results of Algorithm \ref{alg:solve_P3} in Figure \ref{fig:multimode_var_beta} that  we can stop the iterations once an iteration for which
$\eta_M < 0$ (and $\beta_M^\star \in [0, \beta_{M,\max}]$) is reached. Intuitively, if this occurs, only part of the last group receives non-zero power, i.e. is served and it seems unlikely that adding groups with even weaker channels will change that and lead to an increase in the sum rate. 
This can be justified by realizing that 
the limiting sum rate obtained by increasing $\beta_M^\star$ by, say $\beta_\delta$, will be greater or equal to that obtained by adding one more group with group loading $\beta_\delta$. Moreover, increasing $\beta_{M}^\star$ still gives a negative $\eta_M$ which does not satisfy the KKT necessary condition $(\eta_M \geq 0)$.  
Thus,  we can modify Algorithm \ref{alg:solve_P3} by adding a '\textbf{break}' instruction after line $14$. That will stop the iteration and jump directly to line 18.  This will reduce the number of iterations and computations.

\section{Conclusion}\label{s:conclusion_power}
In this paper, we have investigated  problems related to determining the optimal power allocation, regularization parameter and group loadings of a finite number of groups of users so as to maximize the sum rate of MISO broadcast channels with RCI precoder. Even though the analysis was performed in the large system limit, our numerical simulations show its validity for  finite-size system designs. Considering the power allocation problem only, we show that the optimal strategy follows the water-filling scheme. For some cases considered in this paper we show that it is optimal for the BS to communicate to some groups having best channels (highest path-loss gains). We also provide the KKT necessary conditions and propose an algorithm for the optimal group-loading allocation when the BS is allowed to transmit to only subsets of the users in the groups.

\appendices
\section{Proof of Theorem \ref{th:lim_SINR}}\label{proofLSA}
Here, we present the proof briefly since we repeat the same techniques as we have used in deriving the results in \cite{RusdhaPrep, Muharar_asil11,Nguyen_globecom08}. 
An alternative proof of this result is given in \cite[App. II]{Wagner_it12}, the proof corresponding to Theorem 2.
Let $A_k=\frac{1}{N}\mb{h}_k\left(\frac{1}{N}\mb{H}^H_k\mb{H}_k + \rho\mathbf{I}_N\right)^{-1}\mb{h}^H_k=\frac{1}{N}\mb{h}_k\mb{M}_k\mb{h}^H_k$, where $\rho=\alpha/N$, $\mb{M}_k=\left(\frac{1}{N}\mb{H}^H_k\mb{H}_k + \rho\mathbf{I}_N\right)^{-1}$ and $\mb{H}_k$ is $\mb{H}$ with the $k$-th row removed.
Then, by employing the matrix inversion lemma (MIL), $\mb{h}_k(\mb{H}^H\mb{H} + \alpha\mathbf{I}_N)^{-1}\mb{h}^H_k$ in the numerator of \eqref{eq:finite_sinr_power} can be written as $\frac{A_k}{1+A_k}$. 
By using the results \cite[Lemma 1]{Evans_it00} or \cite[Lemma 5.1]{Liang_it07}, it follows that $A_k - \frac{1}{N}\Trf{\mb{M}_k} \as 0$. By using the rank-1 perturbation lemma (R1PL), see e.g. \cite[Theorem 3.9, Lemma 14.3]{Couillet_book11},   $\frac{1}{N}\Trf{\mb{M}_k}$ converges almost surely to $\frac{1}{N}\Trf{\mb{M}}$ with $\mb{M}=(\mb{H}^H\mb{H} + \alpha\mathbf{I}_N)^{-1}$. We can also show that $\frac{1}{N}\Trf{\mb{M}} \as g(\beta,\rho)$ (see the results in e.g. \cite[Theorem 7]{Evans_it00}) where $g(\beta,\rho)$ is the solution of $g(\beta,\rho)=\left(\rho+\frac{\beta}{1+g(\beta,\rho)}\right)^{-1}$. Thus, $A_k - g(\beta,\rho) \as 0$.
Now considering the denominator, we can write $|\mathbf{h}_k(\mathbf{H}^H\mathbf{H}+\alpha \mathbf{I}_N)^{-1}\mathbf{h}^H_j|^2$  as
$\frac{\frac{1}{N}}{(1+A_k)^2}I_j$ where $I_j=\frac{1}{N}\mathbf{h}_k\mb{M}_k\mathbf{h}^H_j\mathbf{h}_j\mb{M}_k\mathbf{h}^H_k$. By \cite[Lemma 5.1]{Liang_it07}, we have 
\[
\max_{j\leq K}\left|I_j-\frac{1}{N}\Trf{\mb{M}_k\mathbf{h}^H_j\mathbf{h}_j\mb{M}_k}\right| \as 0.
\]
The matrix inside the trace has rank one. Thus, the second term on the RHS becomes 
\[
\frac{1}{N}\mb{h}_j\mb{M}_k^{2}\mb{h}_j^H = \frac{1}{(1+A_{j,kj})^2} \frac{1}{N}\mb{h}_j\mb{M}_{kj}^{2}\mb{h}_{j}^H,
\]
where the RHS is obtained by the MIL, $\mb{M}_{kj}=(\frac{1}{N}\mathbf{H}^H_{kj}\mathbf{H}_{kj}+\rho \mathbf{I}_N)^{-1}$, $A_{j,kj}=\frac{1}{N}\mb{h}_j\mb{M}_{kj}\mb{h}_{j}^H$, $\mb{H}_{kj}$ is $\mb{H}_k$ with row $j$ removed. Then, it follows $\max_{j\leq K }\left|\frac{1}{N}\mb{h}_j\mb{M}_{kj}^{2}\mb{h}_j^H-\frac{1}{N}\Trf{\mb{M}_{kj}^{2}}\right|\as 0$. We can show that the second term on the LHS is equal to $-\frac{\partial}{\partial \rho}\frac{1}{N}\Trf{\mb{M}_{kj}}$. 
We also have that $\max_{j\leq K}|A_{j,kj} - \frac{1}{N}\Trf{\mb{M}_{kj}}|\as 0 $ .  By applying R1PL twice,  $\frac{1}{N}\Trf{\mb{M}_{kj}}\as g(\beta,\rho)$.
Suppose that ${\cal P}=\lim_{K\to \infty}\frac{1}{K}\sum_{j=1}^K p_j$ exists and is bounded. Note that it can be interpreted as the empirical mean of the users' power or just average power. Thus, by combining the large system results, we obtain
\be\label{eq:lsa_int_power}
\sum_{j\neq k} p_j |\mathbf{h}_k(\mathbf{H}^H\mathbf{H}+\alpha \mathbf{I}_N)^{-1}\mathbf{h}^H_j|^2 + \beta\mathcal{P} \frac{\frac{\partial g(\beta,\rho)}{\partial \rho}}{(1+g(\beta,\rho))^4} \as 0.
\ee
By following the same steps as in obtaining \eqref{eq:lsa_int_power}, we can establish that
\[
c^2 + \frac{P_d(1+g(\beta,\rho))^2}{\beta {\cal P} \frac{\partial}{\partial \rho}g(\beta,\rho)} \as 0.
\]
Hence, by using this last result, we can conclude that the signal and interference energy converges almost surely to $-P_dp_k a_k^2  g^2(\beta,\rho)\left(\beta\mathcal{P}\frac{\partial}{\partial \rho}g(\beta,\rho)\right)^{-1}$ and $P_d a_k^2(1+g(\beta,\rho))^{-2}$, respectively. 
Recalling the definitions of $\gamma_k$ and $\bar{p}_k$ in the statement of Theorem \ref{th:lim_SINR}, and using the fact that 
\be \label{eq:der_g}
\frac{\partial}{\partial \rho}g(\beta,\rho)=-\frac{g(\beta,\rho)(1+g(\beta,\rho))^2}{\beta+\rho(1+g(\beta,\rho))^2},
\ee 
\eqref{eq:lim_SINR_power} follows immediately. This concludes the proof.

\section{Proof of Theorem \ref{th:opt_rho}}\label{app:opt_rho}
Related to the KKT (stationary) conditions \eqref{eq:stat_rho_power}, for a given $\bar{\mb{p}}$, we have
\[
\frac{\partial R_\text{sum}^\infty}{\partial \rho}=\sum_{j=1}^L \frac{\beta_j \bar{p}_j}{1+\bar{p}_jf_j(\beta,\rho)}\frac{\partial f_j(\beta,\rho)}{\partial \rho},
\]
where
\begin{align*}
\frac{\partial f_j(\beta,\rho)}{\partial \rho}&=\frac{\gamma_j^2}{[\gamma_j+(1+g)^2]^2}2g(1+g)\left(\frac{\rho}{\beta}-\frac{1}{\gamma_j}\right)\frac{\partial g}{\partial \rho} \\
&=f_j^2(\beta,\rho)\frac{2\left(\frac{1}{g}+1\right)}{[1+\frac{\rho}{\beta}(1+g)^2]^2}\left(\frac{\rho}{\beta}-\frac{1}{\gamma_j}\right)\frac{\partial g}{\partial \rho},
\end{align*}
and $g$ represents $g(\beta,\rho)$. Thus, 
\[
\frac{\partial R_\text{sum}^\infty}{\partial \rho}=\frac{2\left(\frac{1}{g}+1\right)}{[1+\frac{\rho}{\beta}(1+g)^2]^2}\sum_{j=1}^L \frac{\beta_j \bar{p}_j f_j^2(\beta,\rho)}{1+\bar{p}_j f_j(\beta,\rho)}\left(\frac{\rho}{\beta}-\frac{1}{\gamma_j}\right)\frac{\partial g}{\partial \rho}.
\]
Recall that $\frac{\partial g}{\partial \rho} < 0$ (see \eqref{eq:der_g}). Let $q_j=\frac{\rho}{\beta}-\frac{1}{\gamma_j}$. It is also obvious that  $q_j$ is decreasing in $j$. Thus, for $q_L > 0$, $\frac{\partial R_\text{sum}^\infty}{\partial \rho}$ is negative. This implies that $\frac{\partial R_\text{sum}^\infty}{\partial \rho}$ can not be zero for $\rho > \frac{\beta}{\gamma_L}$. For $q_1 < 0$,
$\frac{\partial R_\text{sum}^\infty}{\partial \rho}$ is positive and consequently, can not be zero for $\rho < \frac{\beta}{\gamma_1}$. Therefore, the optimal $\rho$ must be in the interval of
\[
\frac{\beta}{\gamma_1} \leq \rho^\star \leq \frac{\beta}{\gamma_L}. 
\]
as in \eqref{eq:interval_opt_rho_power}. When we only have one group then $\rho^*$ is the same as the one obtained in \cite{RusdhaPrep,Nguyen_globecom08}. We can also remove the boundary point $\rho=0 < \frac{\beta}{\gamma_1}$ (related to the case $\kappa > 0$, that is, the constraint $\rho\geq 0$ is inactive) since as previously discussed, $\frac{\partial R_\text{sum}^\infty}{\partial \rho} > 0$ at that point. Thus, from \eqref{eq:stat_rho_power_P1} with $\kappa=0$ or by evaluating  $\frac{\partial R_\text{sum}^\infty}{\partial \rho} = 0$, $\rho^\star$ must satisfy \eqref{eq:opt_rho1_power} at $\bar{\mb{p}}=\bar{\mb{p}}^\star$.

\section{Proof of Lemma \ref{lemma:betaj_alloc}}\label{proofbetaj_alloc}
In the first part, we will prove part (i) - (iii) of the lemma. We show those by considering any two groups $l$ and $j$ such that $l<j$, such that the current allocation has $\beta_j>0$ and $\bar{p}_j>0$ and proving that we can improve performance by having $\beta_l$ at its maximum possible value. Let us assume an assignment $(\beta_l,\bar{p}_l)$ and $(\beta_j,\bar{p}_j)$ such that $\beta_l \leq \beta_{l,\max}$ and $\beta_j \leq \beta_{j,\max}$. In that case, the combined group loading is $\beta_l+\beta_j$. Now, let $x_l$ be the new group loading allocation for group $l$ and $y_l$ be the corresponding assigned power. In the following we will show that the optimal $x_l$ maximizing the sum rate of of users in group $j$ and $l$ is $\beta_{l,\max}$ by solving the following optimization problem
\begin{equation*}
\begin{array}{ccl}
 & \underset{x_l,y_l,y_j}{\text{max.}}  & x_l\log(1+y_lf_l(\beta,\rho)) \\
 & & \qquad\qquad + (\beta_l+\beta_j-x_l)\log(1+y_jf_j(\beta,\rho)) \\
&\text{s.t.} & \max(0,\beta_l+\beta_j-\beta_{j,\max}) \leq x_l,    \\
& & \min(\beta_l+\beta_j,\beta_{l,\max}) \geq x_l, \\
& & y_l x_l + y_j (\beta_l+\beta_j-x_l) \leq \beta_l\bar{p}_l + \beta_j\bar{p}_j,  \\
& & y_l \geq 0, y_j\geq 0.
\end{array}
\end{equation*}
The Lagrangian is given by
\begin{align*}
\mathcal{L}&=x_l\log(1+y_lf_l(\beta,\rho))+(\beta_l+\beta_j-x_l)\log(1+y_jf_j(\beta,\rho)) \\
&\qquad + \mu_{x_l}\left(x_l-\max(0,\beta_l+\beta_j-\beta_{j,\max})\right) \\
&\qquad + \nu_{x_l}\left(\min(\beta_l+\beta_j,\beta_{l,\max}) -x_l \right) \\
&\qquad + \lambda\left(\beta_l\bar{p}_l + \beta_j\bar{p}_j-y_l x_l - y_j (\beta_l+\beta_j-x_l)\right) \\
&\qquad + \mu_{y_l}y_l +\mu_{y_j}y_j,
\end{align*}
where $\mu_{x_l}, \nu_{x_l}, \mu_{y_l}, \mu_{y_j}, \lambda$  are the Lagrange multipliers associated to the constraints on $x_l, y_l, y_j$ and the second constraint, respectively. 
The stationary conditions for the solution candidates are then given by\footnote{Here, we do not use superscript $^\star$ for the solution candidates}
\begin{align}
\frac{\partial \mathcal{L}}{\partial x_l} &= \log(1+y_lf_l(\beta,\rho))  - \log(1+y_jf_j(\beta,\rho)) \notag\\
    & \qquad\qquad\qquad + \mu_{x_l} - \nu_{x_l} - \lambda (y_l - y_j) = 0, \label{eq:stat_xl_P4}\\
\frac{\partial \mathcal{L}}{\partial y_l} &= \frac{x_lf_l(\beta,\rho)}{1+y_l f_l(\beta,\rho)} + \mu_{y_l} - \lambda x_l = 0, \label{eq:stat_yl_P4}\\
\frac{\partial \mathcal{L}}{\partial y_j} &= (\beta_l+\beta_j-x_l)\frac{f_j(\beta,\rho)}{1+y_jf_j(\beta,\rho)} \notag \\
    & \qquad\qquad\qquad  + \mu_{y_j} - \lambda (\beta_l+\beta_j-x_l) = 0. \label{eq:stat_yj_P4}
\end{align}

From \eqref{eq:stat_yl_P4} and \eqref{eq:stat_yj_P4}, it follows that
\begin{align}
y_l &= \left[\frac{1}{\lambda}-\frac{1}{f_l(\beta,\rho)}\right]_+, \label{eq:yl_power}\\ 
y_j &= \left[\frac{1}{\lambda}-\frac{1}{f_j(\beta,\rho)}\right]_+. \label{eq:yj_power}
\end{align}
One can check that $y_l=0$ will never be the optimal solution. For $y_l>0$, two cases arise depending on whether $y_j$ is strictly positive or not.

\begin{itemize}
\item Case $y_j=0$. To satisfy the KKT conditions, the second constraint is met with equality, for $\lambda > 0$. Thus, we have $y_l = \frac{ \beta_l\bar{p}_l + \beta_j\bar{p}_j}{x_l}$. From \eqref{eq:yl_power}, we can express
\[
\frac{1}{\lambda} = \frac{ \beta_l\bar{p}_l + \beta_j\bar{p}_j}{x_l} -\frac{1}{f_l(\beta,\rho)}.
\]
When $y_j =0$, it also holds $1/\lambda - 1/f_j(\beta,\rho) \leq 0$. Consequently, from the equation above, we can write
\[
\frac{\beta_l\bar{p}_l + \beta_j\bar{p}_j}{\frac{1}{f_j(\beta,\rho)}-\frac{1}{f_l(\beta,\rho)}} \leq x_l \ .
\]
From \eqref{eq:stat_xl_P4}, we can obtain
\begin{align}\label{eq:case_yj_zero}
&\log\left(1+\left(\frac{ \beta_l\bar{p}_l + \beta_j\bar{p}_j}{x_l}\right)f_l(\beta,\rho)\right) \notag \\ 
&\qquad\qquad\quad - \frac{1}{1+\frac{x_l}{\beta_l\bar{p}_l + \beta_j\bar{p}_j}\frac{1}{f_l(\beta,\rho)}} =   \nu_{x_l} -\mu_{x_l}.
\end{align}
The LHS of \eqref{eq:case_yj_zero} is a function of the form $f(x)=\log(1+x)-\frac{x}{1+x}$, which can be easily shown to be strictly increasing in $x$. Moreover, at $x=0$, $f(x)=0$. So, the LHS of \eqref{eq:case_yj_zero} is positive. Thus, ignoring the constraint on $x_l$, the objective function is strictly increasing for $\dfrac{\beta_l\bar{p}_l + \beta_j\bar{p}_j}{\frac{1}{f_j(\beta,\rho)}-\frac{1}{f_l(\beta,\rho)}} \leq x_l$.

\vspace{0.5cm}
\item Case $y_j > 0$.  For $\gamma > 0$, the average power constraint is met with equality and we have
\begin{align*}
  y_j &=  \frac{\beta_l\bar{p}_l + \beta_j\bar{p}_j-(y_l-y_j)x_l}{\beta_l + \beta_j} \\
	    &=  \frac{\beta_l\bar{p}_l + \beta_j\bar{p}_j-\left(\frac{1}{f_j(\beta,\rho)}-\frac{1}{f_l(\beta,\rho)}\right)x_l}{\beta_l + \beta_j}.
\end{align*} 
Then, we can express 
\[
\frac{1}{\lambda} = \frac{1}{f_j(\beta,\rho)} + \frac{\beta_l\bar{p}_l + \beta_j\bar{p}_j-\left(\frac{1}{f_j(\beta,\rho)}-\frac{1}{f_l(\beta,\rho)}\right)x_l}{\beta_l + \beta_j}.
\]
Since for $y_j>0$, $\frac{1}{\lambda} > \frac{1}{f_j(\beta,\rho)}$, then we obtain
\be\label{eq:cond_yj_largerzero}
\frac{\beta_l\bar{p}_l + \beta_j\bar{p}_j}{\frac{1}{f_j(\beta,\rho)}-\frac{1}{f_l(\beta,\rho)}} > x_l \ .
\ee
Using the expression for $1/\lambda$, we can rewrite \eqref{eq:stat_xl_P4} as
\begin{align}\label{eq:case_yj_largerzero}
&\nu_{x_l} -\mu_{x_l}=\log\left(\frac{f_l(\beta,\rho)}{f_j(\beta,\rho)}\right) \notag \\
&\quad\quad - \frac{\frac{1}{f_j(\beta,\rho)}-\frac{1}{f_l(\beta,\rho)}}{\frac{1}{f_j(\beta,\rho)} + \frac{\beta_l\bar{p}_l + \beta_j\bar{p}_j-\left(\frac{1}{f_j(\beta,\rho)}-\frac{1}{f_l(\beta,\rho)}\right)x_l}{\beta_l + \beta_j}}.
\end{align}
It is clear that the LHS of \eqref{eq:case_yj_largerzero} is decreasing in $x_l$. Moreover, for $x_l\to\infty$, its value is $\log\left(\frac{f_l(\beta,\rho)}{f_j(\beta,\rho)}\right) > 0$.  Therefore, without the constraints on $x_l$, the objective function is also strictly increasing in $x_l$ when the condition \eqref{eq:cond_yj_largerzero} holds.
\end{itemize}
Combining the two cases, the optimal $x_l$ is equal to its maximum allowable value. By using this fact repeatedly, starting from group 1, we establish (i)-(iii).

Now, it remains to show that if no power is allocated to a group, it must be that the corresponding $\beta_j=0$ (see (iv)).
Let us consider the stationary conditions for $\beta_j$ and $\beta$ which are given by \eqref{eq:stat_betaj_power} and \eqref{eq:stat_beta_power}, respectively. We can rewrite them as
\be\label{eq:stat_betaj_2}
\log(1+\bar{p}_jf_j) - \lambda(\bar{p}_j-1) + \nu_j  + \mu =\eta_j 
\ee
and
\be\label{eq:stat_beta_2}
 \sum_{j=1}^L \frac{\beta_j \bar{p}_j}{1+\bar{p}_jf_j(\beta,\rho)}\frac{\partial f_j(\beta,\rho)}{\partial \beta}  = \mu,
\ee
respectively. In obtaining \eqref{eq:stat_beta_2}, we use the fact that $\beta$ must be positive, i.e., $\eta=0$.
The first derivative of $f_j(\beta,\rho)$ over $\beta$ in \eqref{eq:stat_beta_2} can be shown to take the form
\begin{align}
\frac{\partial f_j(\beta,\rho)}{\partial \beta}&=-\frac{f_j(\beta,\rho)}{\beta}\left[1+\frac{g}{1+\frac{\rho}{\beta} (1+g)^2}\right. \notag \\
&\qquad \left. + \frac{2g(1+g)^2(\frac{\rho}{\beta} \gamma_j-1)}{[\gamma_j+(1+g)^2][1+\frac{\rho}{\beta} (1+g)^2]^2} \right], \label{eq:der_f_beta}
\end{align}
where for brevity we denote $g=g(\beta,\rho)$.
The derivative of $f_j(\beta,\rho)$ w.r.t. $\rho$ in \eqref{eq:stat_rho_power} can  be written as follows
\[
\frac{\partial f_j(\beta,\rho)}{\partial \rho}=-\frac{f_j(\beta,\rho)}{\beta}\frac{2g(1+g)^3(\frac{\rho}{\beta} \gamma_j-1)}{[\gamma_j+(1+g)^2][1+\frac{\rho}{\beta} (1+g)^2]^2}.
\]
So we can rewrite \eqref{eq:der_f_beta} in terms of $\frac{\partial f_j(\beta,\rho)}{\partial \rho}$ as
\begin{align}
\frac{\partial f_j(\beta,\rho)}{\partial \beta} &= -\frac{f_j(\beta,\rho)}{\beta}\left[1+\frac{g}{1+\frac{\rho}{\beta} (1+g)^2}  \right] \notag \\
&\qquad + \frac{1}{1+g}\frac{\partial f_j(\beta,\rho)}{\partial \rho}. \label{eq:der_f_beta_2}
\end{align}
Recall that $1+p_jf_j(\beta,\rho)=f_j(\beta,\rho)/\lambda$. Substituting \eqref{eq:der_f_beta_2} into \eqref{eq:stat_beta_2} yields
\begin{align}\label{eq:mu_at_opt}
\mu &= -\frac{\lambda}{\beta}\left[1+\frac{g}{1+\frac{\rho}{\beta} (1+g)^2}  \right]\sum_{j=1}^L \beta_j \bar{p}_j \notag\\
&\quad + \frac{1}{1+g}\sum_{j=1}^L \frac{\beta_j \bar{p}_j}{1+\bar{p}_jf_j(\beta,\rho)}\frac{\partial f_j(\beta,\rho)}{\partial \rho} \notag\\
& \overset{(a)}{=} -\lambda\left[1+\frac{g}{1+\frac{\rho}{\beta} (1+g)^2}  \right]
\end{align}
where in $(a)$ we use the fact that $\sum_{j=1}^L \beta_j \bar{p}_j = \beta$ and the second term of the RHS is zero due to \eqref{eq:stat_rho_power}. Moreover, $(a)$ gives the expression for $\mu$ at the optimal operating points. Plugging $(a)$ into \eqref{eq:stat_betaj_2} with $p_j=0$, we obtain
 \[
-\lambda\frac{g}{1+\frac{\rho}{\beta} (1+g)^2} +\nu_j = \eta_j.
\]
As a result, $\nu_j$ must be strictly positive. This implies that $\beta_j=0$ and the proof is completed.


%

\ifCLASSOPTIONcaptionsoff
  \newpage
\fi



%
%
\bibliographystyle{IEEEtran}
\bibliography{thesis}

%




\end{document}